\documentclass[a4paper,USenglish,numberwithinsect]{lipics-v2018}

\newif\ifarxiv
\arxivtrue  

\usepackage[utf8]{inputenc}
\usepackage[T1]{fontenc}
\usepackage{amsmath,amsfonts,amssymb,amsthm}
\usepackage{mathrsfs}
\usepackage{graphicx}
\usepackage{hyperref}
\usepackage{url}
\usepackage{color}
\usepackage{cite}
\usepackage{todonotes}
\definecolor{blueblack}{rgb}{0,0,.7}

\newcommand{\emphdef}[1]{%
  \textcolor{blueblack}{%
    \textbf{\emph{#1}}%
  }%
}

\newenvironment{problem}[1]{\smallskip\parindent 0pt{\textcolor{blueblack}{\bf #1}}}{\smallskip} 
\theoremstyle{plain}
\newtheorem{proposition}[theorem]{Proposition}

\newtheorem{theo}{Theorem}[section]

\newtheorem{prop}[theo]{Proposition}

\newcommand{\NN}{\mathbb{N}}
\newcommand{\R}{\mathbb{R}}
\newcommand{\SSS}{\mathbb{S}}

 \newcommand{\complex}{\mathscr{C}}
\newcommand {\surf} {\mathscr{S}}

\newcommand\Input{\textsc{Input}}

\newcommand\Question{\textsc{Question}}
\newcommand\Embed{\textsc{Embed}}
\newcommand\EEP{\textsc{EEP}}
\newcommand\EEPSing{\textsc{EEP-Sing}}
\newcommand\EEPSurf{\textsc{EEP-Surf}}
\newcommand\EEPCon{\textsc{EEP-Conn}}
\newcommand\EEPCell{\textsc{EEP-Cell}}
\newcommand\EmbedResp{\textsc{Embed-Resp}}
\newcommand\GraphGenus{\textsc{Graph-Genus}}

\title{Embedding Graphs into Two-Dimensional Simplicial Complexes}
\funding{Part of this work was done while the first two authors were at D\'epartement d'Infor\-matique, \'Ecole normale sup\'erieure, Paris, France, and while the third author was invited professor there.  
The first two authors are supported by grant ANR-17-CE40-0033 of the French National Research Agency ANR (SoS project).}

\ifarxiv\else
\relatedversion{A full version of this paper is available at \url{XXX}}
\fi

\author{\'Eric Colin de Verdi\`ere}{Universit\'e Paris-Est, LIGM, CNRS, ENPC, ESIEE Paris, UPEM, Marne-la-Vall\'ee\\{France}}{eric.colindeverdiere@u-pem.fr}{}{}
\author{Thomas Magnard}{Universit\'e Paris-Est, LIGM, CNRS, ENPC, ESIEE Paris, UPEM, Marne-la-Vall\'ee\\{France}}{thomas.magnard@u-pem.fr}{}{}
\author{Bojan Mohar}{Department of Mathematics, Simon Fraser University\\{Burnaby, Canada}}{mohar@sfu.ca}{}{}

\authorrunning{\'E.~Colin de Verdi\`ere, T.~Magnard, and B.~Mohar}

\ifarxiv
\hideLIPIcs
\else\Copyright{\'Eric Colin de Verdi\`ere, Thomas Magnard, and Bojan Mohar}

\EventEditors{Bettina Speckmann and Csaba D. T{\'o}th}
\EventNoEds{2}
\EventLongTitle{34th International Symposium on Computational Geometry (SoCG 2018)}
\EventShortTitle{SoCG 2018}
\EventAcronym{SoCG}
\EventYear{2018}
\EventDate{June 11--14, 2018}
\EventLocation{Budapest, Hungary}
\EventLogo{socg-logo} 
\SeriesVolume{99}
\ArticleNo{} %

\fi

\def\ra{$\rightarrow$ }

\subjclass{Theory of computation \ra Randomness, geometry and discrete structures \ra Computational geometry; Mathematics of computing \ra  Discrete mathematics \ra  Graph theory \ra  Graph algorithms; 
Mathematics of computing \ra  Discrete mathematics \ra  Graph theory \ra  Graphs and surfaces; Mathematics of computing \ra  Continuous mathematics \ra  Topology \ra Algebraic topology}


\nolinenumbers

\begin{document}
\maketitle

\begin{abstract}
 We consider the problem of deciding whether an input graph~$G$ admits a topological embedding into a two-dimensional simplicial complex~$\complex$. This problem includes, among others, the embeddability problem of a graph on a surface and the topological crossing number of a graph, but is more general.

The problem is NP-complete when $\complex$ is part of the input, and we give a polynomial-time algorithm if the complex $\complex$ is fixed.

Our strategy is to reduce the problem to an embedding extension problem on a surface, which has the following form: Given a subgraph $H'$ of a graph $G'$, and an embedding of $H'$ on a surface $S$, can that embedding be extended to an embedding of $G'$ on $S$? Such problems can be solved, in turn, using a key component in Mohar's algorithm to decide the embeddability of a graph on a fixed surface (STOC 1996, SIAM J.~Discr.\ Math.\ 1999). 
\end{abstract}

\ifarxiv

\bigskip

This is the version with Appendix of a paper to appear in \emph{Proceedings of the 34th International Symposium on Computational Geometry} (SoCG 2018).

\bigskip\bigskip
\fi

\section{Introduction}
\subparagraph*{Topological embedding problems.}
Topological embedding problems are among the most fundamental problems in computational topology, already emphasized since the early developments of this discipline~\cite[Section~10]{deg-ct-99}.  Their general form is as follows: Given topological spaces $X$ and~$Y$, does there exist an embedding (a continuous, injective map) from~$X$ to~$Y$?  Since a finite description of $X$ and~$Y$ is needed, typically they are represented as finite simplicial complexes, which are topological spaces obtained by attaching simplices (points, segments, triangles, tetrahedra, etc.) of various dimensions together.

The case where the host space~$Y$ equals~$\R^d$ (or, almost equivalently, $\SSS^d$, which can be modeled as a simplicial complex) has been studied the most.  The case $d=2$ corresponds to the planarity testing problem, which has attracted considerable interest~\cite{p-pte-06}.  The case $d=3$ is much harder, and has only recently been shown to be decidable by Matou\v{s}ek, Sedgwick, Tancer, and Wagner~\cite{mstw-e3sd-18}.  The general problem for arbitrary~$d$ has been extensively studied in the last few years, starting with hardness results by Matou\v{s}ek, Tancer, and Wagner~\cite{mtw-hescr-11}, and continuing with some algorithmic results in a series of articles; we refer to Matou\v{s}ek et al.~\cite[Introduction]{mstw-e3sd-18} for a state of the art.

What about more general choices of~$Y$?  The case where $Y$ is a graph is essentially the subgraph homeomorphism problem, asking if $Y$~contains a subdivision of a graph~$X$.  This is hard in general, easy when $Y$~is fixed, and polynomial-time solvable for every fixed~$X$, by using graph minor algorithms.  The case where $X$ is a graph and~$Y$ a 2-dimensional simplicial complex that is homeomorphic to a surface has been much investigated, also in connection to topological graph theory~\cite{mt-gs-01} and algorithms for surface-embedded graphs~\cite{e-cocb-12,dmst-apgb-11}: The problem is NP-complete, as proved by Thomassen~\cite{t-ggpnc-89}, but Mohar~\cite{m-ltaeg-99} has proved that it can be solved in linear time if $Y$ is fixed (in some recent works, the proof has been simplified and the result extended~\cite{kmr-sltae-08,kp-dvgbg-17}).  The case where $X$ is a 2-complex and $Y$ is (a 2-complex homeomorphic to) a surface essentially boils down to the previous case; see Mohar~\cite{m-mg2c-97}.  More recently, {\v C}adek, Kr{\v c}{\'a}l, Matou{\v s}ek, Vok{\v r}{\'\i}nek, and Wagner~\cite[Theorem~1.4]{ckmvw-ptchg-14} considered the case where the host complex~$Y$ has an arbitrary (but fixed) dimension; they provide a polynomial-time algorithm for the related \emph{map extension problem}, 
under some assumptions on the dimensions of $X$ and~$Y$; in particular, $Y$ must have trivial fundamental group (because they manipulate in an essential way the homotopy groups of~$Y$, which have to be Abelian); but the maps they consider need not be embeddings.

Another variation on this problem is to try to embed $X$ such that it extends a given partial embedding of $X$ (we shall consider such \emph{embedding extension problems} later). This problem has already been studied in some particular cases; in particular, Angelini, Battista, Frati, Jel{\'i}nek, Kratochv{\`\i}l, Patrignani, and Rutter~\cite[Theorem 4.5]{angelini2015testing} provide a linear-time algorithm to decide the embedding extension problem of a graph in the plane.

\subparagraph*{Our results.}
In this article, we study the topological embedding problem when $X$ is an arbitrary graph~$G$, and $Y$ is an arbitrary two-dimensional simplicial complex~$\complex$ (actually, a simplicial complex of dimension at most two---abbreviated as \emph{2-complex} below).  Formally, we consider the following decision problem:

\begin{problem}
\Embed$(n,c)$:\\
\Input: A graph~$G$ with $n$ vertices and edges, and a 2-complex~$\complex$ with $c$ simplices.\\
\Question: Does $G$ have a topological embedding into~$\complex$?
\end{problem}

\noindent(We use the parameters $n$ and~$c$ whenever we need to refer to the input size.)  Here are our main results:
\begin{theo}\label{T:np}
  The problem \Embed{} is NP-complete.
\end{theo}
\begin{theo}\label{T:main}
  The problem \Embed$(n,c)$ can be solved in time $f(c)\cdot n^{O(c)}$, where $f$ is some computable function of~$c$.
\end{theo}

As for Theorem~\ref{T:np}, it is straightforward that the problem is NP-hard (as the case where $\complex$ is a surface is already NP-hard); the interesting part is to provide a certificate checkable in polynomial time when an embedding exists.  Note that Theorem~\ref{T:main} shows that, for every fixed complex~$\complex$, the problem of deciding whether an input graph embeds into~$\complex$ is polynomial-time solvable.  Actually, our algorithm is explicit, in the sense that, if there exists an embedding of~$G$ on~$\complex$, we can provide some representation of such an embedding (in contrast to some results in the theory of graph minors, where the existence of an embedding can be obtained without leading to an explicit construction).

\subparagraph*{Why do 2-complexes look harder than surfaces?}

A key property of the class of graphs embeddable on a fixed surface is that it is minor-closed: Having a graph~$G$ embeddable on a surface~$\surf$, removing or contracting any edge yields a graph embeddable on~$\surf$.  By Robertson and Seymour's theory, this immediately implies a cubic-time algorithm to test whether a graph~$G$ embeds on~$\surf$, for every fixed surface~$\surf$~\cite{rs-gm13d-95}.  In contrast, the class of graphs embeddable on a fixed 2-complex is, in general, not closed under taking minors, and thus this theory does not apply.  For example, let $\complex$ be obtained from two tori by connecting them together with a line segment, and let $G$ be obtained from two copies of~$K_5$ by joining them together with a new edge~$e$; then $G$ embeds into~$\complex$, but the minor obtained from~$G$ by contracting~$e$ does not.

Two-dimensional simplicial complexes are topologically much more complicated than surfaces.  For example, there exist linear-time algorithms to decide whether two surfaces are homeomorphic (this amounts to comparing the Euler characteristics, the orientability characters, and, in case of surfaces with boundary, the numbers of boundary components), or to decide whether a closed curve is contractible (see Dey and Guha~\cite{dg-tcs-99}, Lazarus and Rivaud~\cite{lr-hts-12}, and Erickson and Whittlesey~\cite{ew-tcsr-13}).  In contrast, the homeomorphism problem for 2-complexes is as hard as graph isomorphism, as shown by \'O~D{\'u}nlaing, Watt, and Wilkins~\cite{oww-h2ceg-00}.  Moreover, the contractibility problem for closed curves on 2-complexes is undecidable; even worse, there exists a fixed 2-complex~$\complex$ such that the contractibility problem for closed curves on~$\complex$ is undecidable (this is because every finitely presented group can be realized as the fundamental group of a 2-complex, and there is such a group in which the word problem is undecidable, by a result of Boone~\cite{b-wp-59}; see also Stillwell~\cite[Section~9.3]{s-ctcgt-93}).

Despite this stark contrast between surfaces and 2-complexes, if we care only on the polynomiality or non-polynomiality, our results show that the complexities of embedding a graph into a surface or a 2-complex are similar:  If the host space is not fixed, the problem is NP-complete; otherwise, it is polynomial-time solvable.  Compared to the aforementioned hard problems on general 2-complexes, one feature related to our result is that every graph embeds on a 3-book (a complex made of three triangles sharing a common edge); thus, we only need to consider 2-complexes without 3-book, for otherwise the problem is trivial; this significantly restricts the structure of the 2-complexes to be considered.  The problem of whether \Embed{} admits an algorithm that is fixed-parameter tractable in terms of the parameter~$c$, however, remains open for general complexes, whereas it is the case when restricting to surfaces~\cite{m-ltaeg-99}.

\subparagraph*{Why is embedding graphs on 2-complexes interesting?}

First, let us remark that, if we consider the problem of embedding graphs into simplicial complexes, then the case that we consider, in which the complex has dimension at most two, is the only interesting one, since every graph can be embedded in a single tetrahedron.

We have already noted that the problem we study is more general than the problem of embedding graphs on surfaces.  It is indeed quite general, and some other problems studied in the past can be recast as an instance of \Embed{} or as variants of it.
For example, the \emph{crossing number} of a graph~$G$ is the minimum number of crossings in a (topological) drawing of~$G$ in the plane.  Deciding whether a graph~$G$ has crossing number at most~$k$ is NP-hard, but fixed-parameter tractable in~$k$, as shown by Kawarabayashi and Reed~\cite{kr-ccnlt-07}.  This is easily seen to be equivalent to the embeddability of~$G$ into the complex obtained by removing $k$~disjoint disks from a sphere and adding, for each resulting boundary component~$b$, two edges with endpoints on~$b$ whose cyclic order along~$b$ is interlaced.  Of course, the embeddability problem on a 2-complex is more general and contains, for example, the problem of deciding whether there is a drawing of a graph~$G$ on a surface of genus~$g$ with at most $k$ crossings.  In topological graph theory, embeddings of graphs on pseudosurfaces (which are special 2-complexes) have been considered; see Archdeacon~\cite[Section~5.7]{a-tgts-96} for a survey.
Slightly more remotely, a \emph{book embedding} of a graph~$G$ (see, e.g., Malik~\cite{m-ggghp-94}) is also an embedding of~$G$ into a particular 2-complex, with additional constraints on the embedding.

\subparagraph*{Strategy of the proof and organization of the paper.}

For clarity of exposition, in most of the paper, we focus on developing an algorithm for the problem \Embed{} (Theorem~\ref{T:main}).  Only at the end (Section~\ref{S:main_proof}) we explain why our techniques imply that the problem is in NP.  The idea of the algorithm is to progressively reduce the problem to simpler problems.  We first deal with the case where the complex~$\complex$ contains a 3-book (Section~\ref{S:3-book}).  From Section~\ref{S:impure} onwards, we reduce \Embed{} to \emph{embedding extension problems} (EEP), similar to the \Embed{} problem except that an embedding of a subgraph~$H$ of the input graph~$G$ is already specified.  In Section~\ref{S:impure}, we reduce \Embed{} to EEPs on a pure 2-complex (in which every segment of the complex~$\complex$ is incident to at least one triangle).  In Section~\ref{S:pure}, we further reduce it to EEPs on a surface.  In Section~\ref{S:cellularise}, we reduce it to EEPs on a surface in which every face of the subgraph~$H$ is a disk.  Finally, in Section~\ref{S:cellular_solving}, we show how to solve EEPs of the latter type using a key component in an algorithm by the third author~\cite{m-ltaeg-99} to decide embeddability of a graph on a surface.


\section{Preliminaries}\label{S:prelims}
\subsection{Embeddings of graphs into 2-complexes}

A \emphdef{2-complex} is an abstract simplicial complex of dimension at most two: a finite set of 0-simplices called \emphdef{nodes}, 1-simplices called \emphdef{segments}, and 2-simplices called \emphdef{triangles} (we use this terminology to distinguish from that of vertices and edges, which we reserve for graphs); each segment is a pair of nodes, and each triangle is a triple of nodes; moreover, each subset of size two in a triangle must be a segment.  Each 2-complex~$\complex$ corresponds naturally to a topological space, obtained in the obvious way: Start with one point per node in~$\complex$; connect them by segments as indicated by the segments in~$\complex$; similarly, for every triangle in~$\complex$, create a triangle whose boundary is made of the three segments contained in that triangle.  By abuse of language, we identify $\complex$ with that topological space.  To emphasize that we consider the abstract simplicial complex and not only the topological space, we sometimes use the name \emphdef{triangulation} or \emphdef{triangulated complex}.

In this paper, graphs are finite, undirected, and may have loops and multiple edges.  In a similar way as for 2-complexes, each graph has an associated topological space; an \emphdef{embedding} of a graph~$G$ into a 2-complex~$\complex$ is an injective continuous map from (the topological space associated to) $G$ to (the topological space associated to)~$\complex$.

\subsection{Structural aspects of 2-complexes}

We say that a 2-complex \emphdef{contains a 3-book} if some three distinct triangles share a common segment.

Let $p$ be a node of~$\complex$.
A \emphdef{cone at~$p$} is a cyclic sequence of triangles $t_1,\ldots,t_k,t_{k+1}=t_1$, all incident to~$p$, such that, for each $i=1,\ldots,k$, the triangles $t_i$ and~$t_{i+1}$ share a segment incident with~$p$, and any other pair of triangles have only $p$ in common.  A \emphdef{corner at~$p$} is an inclusionwise maximal sequence of triangles $t_1,\ldots,t_k$, all incident to~$p$, such that, for each~$i=1,\ldots,k-1$, the triangles $t_i$ and~$t_{i+1}$ share a segment incident with~$p$, and any other pair of triangles have only $p$ in common.  An \emphdef{isolated segment at~$p$} is a segment incident to~$p$ but not incident to any triangle.

If $\complex$ contains no 3-books, the set of segments and triangles incident with a given node~$p$ of~$\complex$ are uniquely partitioned into cones, corners, and isolated segments.  We say that $p$ is a \emphdef{regular node} if all the segments and triangles incident to~$p$ form a single cone or corner.  Otherwise, $p$ is a \emphdef{singular node}.
A 2-complex is \emphdef{pure} if it contains no isolated segment, and each node is incident to at least one segment.

\subsection{Embedding extension problems and reductions}

An \emphdef{embedding extension problem} (EEP) is a decision problem defined as follows:

\begin{problem}
  \EEP$(n,m,c)$:\\
  \Input: A graph~$G$ with $n$ vertices and edges, a subgraph~$H$ of~$G$ with $m$ vertices and edges, and an embedding~$\Pi$ of~$H$ into a 2-complex~$\complex$ with $c$ simplices.\\
  \Question: Does $G$ have an embedding into~$\complex$ whose restriction to~$H$ is~$\Pi$?
\end{problem}

To be precise, we will have to explain how we represent the embedding~$\Pi$, but this will vary throughout the proof, and we will be more precise about this in subsequent sections.  Let us simply remark that, since the complexity of our algorithm is a polynomial of large degree (depending on the complex~$\complex$) in the size of the input graph, the choice of representation is not very important, because converting between any two reasonable representations is possible in polynomial time.

We will reduce our original problem to more and more specialized EEPs.  We will use the word ``reduce'' in a somewhat sloppy sense:  A decision problem~$P$ \emphdef{reduces} to $k$ instances of the decision problem~$P'$ if solving these $k$ instances of~$P'$ allows to solve the instance of~$P$ in time $O(k)$.  We will have to be more precise when we consider the NP-completeness of \Embed{} in Section~\ref{S:main_proof}.

\subsection{Surfaces}

In Section~\ref{S:cellularise}, we will assume some familiarity with surface topology; see, e.g., \cite{mt-gs-01,s-ctcgt-93,c-ctgs-17} for suitable introductions under various viewpoints.  We recall some basic definitions and properties.  A \emphdef{surface} $\surf$ is a compact, connected Hausdorff topological space in which every point has a neighborhood homeomorphic to the plane.  Every surface~$\surf$ is obtained from a sphere by:
\begin{itemize}
    \item either removing $g/2$ open disks and attaching a handle (a torus with an open disk removed) to each resulting boundary component, for an even, nonnegative number~$g$ called the \emph{(Euler)} \emphdef{genus} of~$\surf$; in this case, $\surf$ is \emphdef{orientable};
    \item or removing $g$ open disks and attaching a M\"obius band to each resulting boundary component, for a positive number~$g$ called the \emphdef{genus} of~$\surf$; in this case, $\surf$ is \emphdef{non-orientable}.
\end{itemize}
A \emphdef{surface with boundary} is obtained from a surface by removing a finite set of interiors of disjoint closed disks.  The boundary of each disk forms a \emphdef{boundary component} of~$\surf$.  A \emphdef{possibly disconnected surface} is a disjoint union of surfaces.
An embedding of~$G$ into a surface~$\surf$, possibly with boundary, is \emphdef{cellular} if each face of the embedding is homeomorphic to an open disk.  If $G$ is cellularly embedded on a surface with genus~$g$ and $b$ boundary components, with $v$~vertices, $e$~edges, and $f$~faces, then Euler's formula stipulates that $2-g-b=v-e+f$.

An \emphdef{ambient isotopy} of a surface with boundary~$\surf$ is a continuous family $(h_t)_{t\in[0,1]}$ of self-homeomorphisms of~$\surf$ such that $h_0$ is the identity.

\section{Reduction to complexes containing no 3-book}\label{S:3-book}
The following folklore observation allows us to solve the problem trivially if~$\complex$ contains a 3-book.  We include a proof for completeness.
\begin{figure}
\centering
\includegraphics[width=\linewidth]{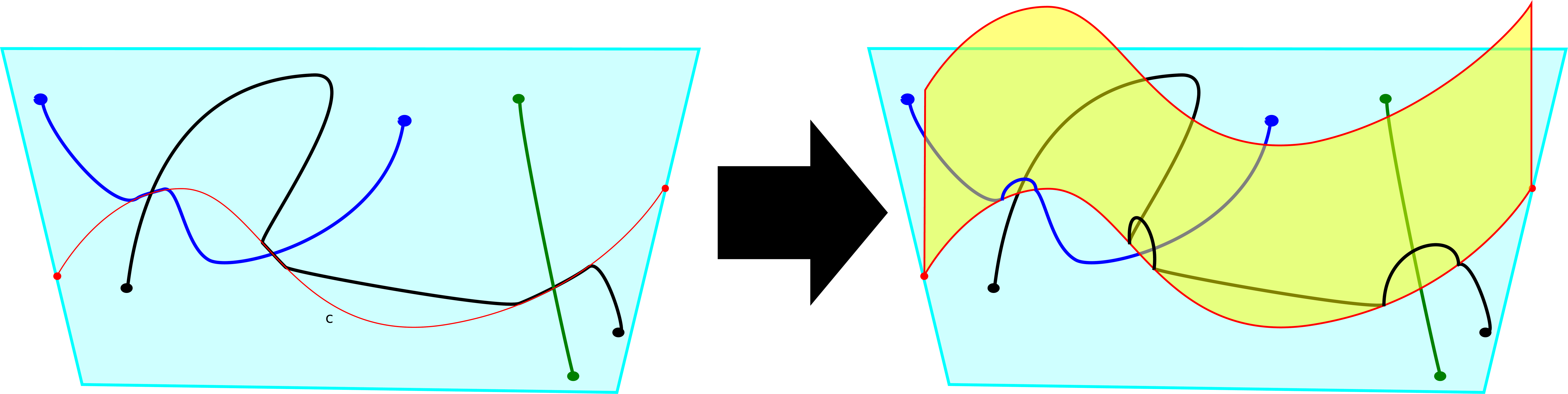}
\caption{Illustration of the proof of
  Proposition~\ref{P:3-book}. Left: The drawing of the graph~$G$ and the curve~$c$ (in thin).   Right: The construction of the 3-book and the modification of the drawing.}
\label{F:3-book}
\end{figure}
\begin{prop}\label{P:3-book}
  If $\complex$ contains a 3-book, then every graph embeds into~$\complex$.
\end{prop}
\begin{proof}
  Let $G$ be a graph.  We first draw $G$, possibly with crossings, in general position in the interior of a closed disk~$D$.  Let $c$ be a simple curve in~$D$ with endpoints on $\partial D$ and passing through all crossing points of the drawing of~$G$.  By perturbing $c$, we can ensure that, in the neighborhood of each crossing point of that drawing, $c$ coincides with the image of one of the two edges involved in the crossing.  See Figure~\ref{F:3-book}, left. 
  
  Let $D'$ be a closed disk disjoint from~$D$.  We attach $D'$ to~$D$ by identifying $c$ with a part of the boundary of~$D'$. Now, in the neighborhood of each crossing of the drawing of~$G$, we push  inside~$D'$ the part of the edge coinciding with~$c$, keeping its endpoints fixed.  See Figure~\ref{F:3-book}, right.  This removes the crossings.
  
  So $G$ embeds in the topological space obtained from~$D$ by attaching a part of the boundary of~$D'$ along~$c$.  But this space embeds in~$\complex$, because $\complex$ contains a 3-book.
 \end{proof}

\section{Reduction to EEPs on a pure 2-complex}\label{S:impure}

Our next task is to reduce the problem \Embed{} to a problem on a pure 2-dimensional complex.  More precisely, let \emphdef{\EEPSing} be the problem \EEP, restricted to instances $(G,H,\Pi,\complex)$ where: $\complex$ is a pure 2-complex containing no 3-books; $H$ is a set of vertices of~$G$; and $\Pi$ is an injective map from $H$ to the nodes of~$\complex$ such that $\Pi(H)$ contains all singular nodes of~$\complex$.  In this section, we prove the following result.

\begin{proposition}\label{P:pure}
  Any instance of \Embed$(n,c)$ reduces to $(cn)^{O(c)}$ instances of \EEPSing$(cn,c,O(c))$.
\end{proposition}

First, a definition.  Consider a map $f:P\to V(G)\cup\{\varepsilon\}$, where $P$ is a set of nodes in~$\complex$ containing all singular nodes of~$\complex$.   We say that an embedding~$\Gamma$ of~$G$ \emphdef{respects}~$f$ if, for each $p\in P$, the following holds:
If $f(p)=\varepsilon$, then $p$ is not in the image of~$\Gamma$; otherwise, $\Gamma(f(p))=p$.

In this section, we will need the following intermediate problem:

\begin{problem}\EmbedResp$(n,m,c)$:\\
\Input: A graph~$G$ with $n$ vertices and edges, a 2-complex~$\complex$ (not necessarily pure) containing no 3-books, with $c$ simplices, and a map~$f$ as above, with domain of size~$m$.\\
\Question: Does $G$ have an embedding into~$\complex$ respecting~$f$?
\end{problem}

\begin{lemma}\label{L:guess-pure}
  Any instance of \Embed$(n,c)$ reduces to $(O(cn))^c$ instances of\\ \EmbedResp$(cn,c,c)$.
\end{lemma}
\begin{proof}
  By Proposition~\ref{P:3-book}, we can without loss of generality assume that $\complex$ contains no 3-books.  Let $G'$ be the graph obtained from $G$ by subdividing each edge $k$~times, where $k\le c$ is the number of singular nodes of~$\complex$. We claim that $G$ has an embedding into~$\complex$ if and only if $G'$ has an embedding~$\Gamma'$ into~$\complex$ such that each singular node of~$\complex$ in the image of~$\Gamma'$ is the image of a vertex of~$G'$.

  Indeed, assume that $G$ has an embedding~$\Gamma$ on~$\complex$.  Each time an edge of~$G$ is mapped, under~$\Gamma$, to a singular node~$p$ of~$\complex$, we subdivide this edge and map this new vertex to~$p$; the image of the embedding is unchanged.  This ensures that only vertices are mapped to singular nodes.  Moreover, there were at most $k$ subdivisions, one per singular node. So, by further subdividing the edges until each original edge is subdivided $k$ times, we obtain an embedding of~$G'$ to~$\complex$ such that only vertices are mapped on singular nodes.  The reverse implication is obvious: If $G'$ has an embedding into~$\complex$, then so has~$G$.  This proves the claim.
  
  To conclude, for each map from the set of singular vertices of~$\complex$ to $V(G')\cup\{\varepsilon\}$, we solve the problem whether $G'$ has an embedding on~$\complex$ respecting~$f$.  The graph~$G$ embeds on~$\complex$ if and only if the outcome is positive for at least one such map~$f$.   By construction, there are at most $(kn+1)^k=(O(nc))^c$ such maps, because $V(G')$ has size at most~$kn$.
\end{proof}

\begin{lemma}\label{L:reduc-pure}
  \EmbedResp$(n,m,c)$ reduces to \EEPSing$(n,m,O(c))$.
\end{lemma}
\begin{proof}
\ifarxiv 
The proof is a bit long, and we refer to Appendix~\ref{A:impure}, but the idea is simple:  
\fi
Because we restrict ourselves to embeddings that respect~$f$, and thus specify which vertex of~$G$ is mapped to each of the singular nodes of~$\complex$, what happens on the isolated segments of~$\complex$ is essentially determined.
\end{proof}

Before giving the proof, we first note:
\begin{proof}[Proof of Proposition~\ref{P:pure}]
  It follows immediately from Lemmas~\ref{L:guess-pure} and~\ref{L:reduc-pure}.
\end{proof}

\section{Reduction to an EEP on a possibly disconnected surface}\label{S:pure}
The previous section led us to an embedding extension problem in a pure 2-complex without 3-book where the images of some vertices are predetermined.  Now, we show that solving such an EEP amounts to solving another EEP in which the complex is a surface.  

Let \emphdef{\EEPSurf} be the problem \EEP, restricted to instances where the input complex is (homeomorphic to) a possibly disconnected triangulated surface without boundary (which we denote by~$\surf$ instead of~$\complex$, for clarity).  To represent the embedding~$\Pi$ in such an EEP instance $(G,H,\Pi,\surf)$, it will be convenient to use the fact that, in all our constructions below, the image of every connected component of~$H$ under~$\Pi$ will intersect the 1-skeleton of~$\surf$ at least once, and finitely many times.  (Note that $H$ may use some nodes of~$\surf$.)  Consider the \emph{overlay} of the triangulation of~$\surf$ and of~$\Pi$, the union of the 1-skeleton of~$\surf$ and of the image of~$\Pi$; this overlay is the image of a graph on~$\surf$; each of its edges is either a piece of the image of an edge of~$H$ or a piece of a segment of~$\surf$; each of its vertices is the image of a vertex of~$H$ and/or a node of~$\surf$.  By the assumption above on~$\Pi$, this overlay is cellularly embedded on~$\surf$, and we can represent it by its combinatorial map~\cite{l-gem-82,e-dgteg-03} (possibly on surfaces with boundary, since at intermediary steps of our construction we will have to consider such surfaces).

In this section, we prove the following proposition.
\begin{proposition}\label{P:surf}
 Any instance of \EEPSing$(n,m,c)$ reduces to an instance of \EEPSurf$(O(n+m+c),O(m+c),O(c))$.
\end{proposition}

We will first reduce the original EEP to an intermediary EEP on a surface with boundary.

\begin{lemma}\label{L:surf-bound}
Any instance of \EEPSing$(n,m,c)$ reduces to an instance of \EEP$(n+O(c),m+O(c),O(c))$ in which the considered 2-complex is a possibly disconnected surface with boundary.
\end{lemma}

\begin{figure}
\centering
\includegraphics[width=.8\linewidth]{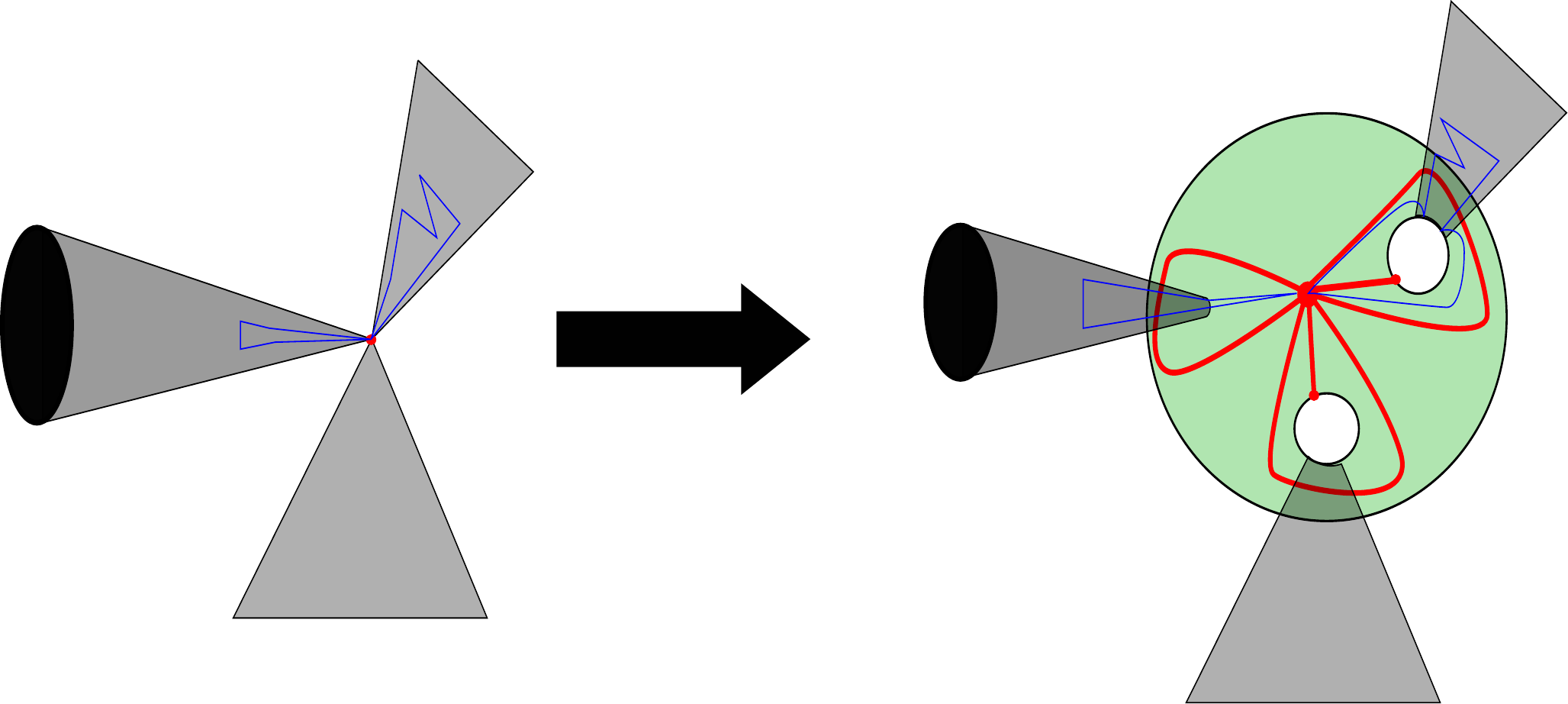}
\caption{The modification of singular vertices in the proof of
  Lemma~\ref{L:surf-bound}.  We transform the neighborhood of
  each singular vertex to make it surface-like.  Moreover,
  we add to~$H$ one loop per cone or corner; furthermore, for each corner, we add a vertex and an edge.}
\label{F:thickening}
\end{figure}

\begin{proof}
  The key property that we will use is that, since $\complex$ is pure and contains no 3-books, each singular node is incident to cones and corners only.
  
  Figure~\ref{F:thickening} illustrates the proof.  Let $(G,H,\Pi,\complex)$ be the instance of~\EEPSing.  We first describe the construction of the instance $(G',H',\Pi',\surf)$ on the possibly disconnected surface with boundary.  Let $p$ be a singular node; we modify the complex in the neighborhood of~$p$ as follows.  Let $c_p$ be the number of cones at~$p$ and $c'_p$ be the number of corners at~$p$.  We remove a small open neighborhood~$N_p$ of~$p$ from~$\complex$, in such a way that the boundary of~$N_p$ is a disjoint union of $c_p$ circles and $c'_p$ arcs.  We create a sphere~$S_p$ with $c_p+c'_p$ boundary components.  Finally, we attach each circle and arc to a different boundary component of~$S_p$, choosing an arbitrary orientation for each gluing; circles are attached to an entire boundary component of~$S_p$, while arcs cover only a part of a boundary component of~$S_p$.  Doing this for every singular node~$p$, we obtain a surface (possibly disconnected, possibly with boundary), which we denote by~$\surf$. 

  We now define $H'$, $G'$, and~$\Pi'$ from $H$, $G$, and~$\Pi$ (again, refer to Figure~\ref{F:thickening}).  Let $p$ be a singular node of~$\complex$ and $v_p$ the vertex of~$H$ mapped on~$p$ by~$\Pi$.  In $H$ (and thus also~$G$), we add a set~$L_p$ of $c_p+c'_p$ loops with vertex~$v_p$.  Let $q_p$ be a point in the interior of $S_p$; in~$\Pi$, we map $v_p$ to~$q_p$, and we map these $c_p+c'_p$ loops on~$S_p$ in such a way that each loop encloses a different boundary component of~$S_p$ (thus, if we cut $S_p$ along these loops, we obtain $c_p+c'_p$ annuli and one disk).

  Finally, we add to~$H$ (and thus also to~$G$) a set~$E_p$ of $c'_p$ new edges, each connecting $v_p$ to a new vertex.  In $\Pi$, each new vertex is mapped to the boundary component of~$S_p$ corresponding to a corner, but not on the corresponding arc.

  Let us call $G'$ and~$H'$ the resulting graphs, and~$\Pi'$ the resulting embedding of~$H'$.  Note that, from the triangulation of~$\complex$ with $c$ simplices, we can easily obtain a triangulation of~$\surf$ with $O(c)$ simplices, and that the image of each edge of~$H$ crosses $O(1)$ edges of this triangulation.
  
  \ifarxiv There remains to prove that these two \EEP{}s are equivalent; see Appendix~\ref{A:pure}.
  \fi
\end{proof}

We now deduce from the previous EEP the desired EEP on a surface without boundary.
\begin{lemma}\label{L:surf nobound}
 Any instance of \EEPSing$(n,m,c)$ on a possibly disconnected surface with boundary reduces to an instance of \EEPSurf$(n+O(m),O(m),O(c))$.
\end{lemma}
\begin{proof}
  \ifarxiv (Figure~\ref{F:remove-boundary}, in Appendix, illustrates the proof.) \fi
  Let $(G,H,\Pi,\surf)$ be an instance of an EEP on a possibly disconnected surface with boundary.  We first describe the construction of $(G',H',\Pi',\surf')$, the EEP instance on a possibly disconnected surface without boundary.
  
  Let $\surf'$ be obtained from~$\surf$ by gluing a disk $D_b$ along each  boundary component.  Let $b$ be a boundary component of $\surf$.  If $\Pi$ maps at least one vertex to~$b$, then we add to $H$ (and thus also to $G$) a new vertex~$v_b$, which we connect, also by a new edge, to each of the vertices mapped to~$b$ by~$\Pi$.  We extend~$\Pi$ by mapping vertex~$v_b$ and its incident edges inside~$D_b$.  Let us call $G'$ and~$H'$ the resulting graphs, and~$\Pi'$ the resulting embedding of~$H'$. For each vertex of~$H$ on a boundary component, we added to $H$ and $G$ at most one vertex and one edge.
  \ifarxiv There remains to prove that these two \EEP{}s are equivalent; see Appendix~\ref{A:pure}.
  \fi
\end{proof}
Finally:
\begin{proof}[Proof of Proposition~\ref{P:surf}]
It suffices to successively apply  Lemmas \ref{L:surf-bound} and \ref{L:surf nobound}. 
\end{proof}

\section{Reduction to a cellular EEP on a surface}\label{S:cellularise}
Let \emphdef{\EEPCell} be the problem \EEP, restricted to instances $(G,H,\Pi,\surf)$ where $\surf$ is a surface and $H$ is cellularly embedded and intersects each connected component of~$G$.

In this section, we prove the following proposition.

\begin{proposition}\label{P:cell}
  Any instance of \EEPSurf$(n,m,c)$ reduces to $(n+m+c)^{O(m+c)}$ instances of \EEPCell$(O(n+m+c),O(n+m+c),c)$.
\end{proposition}

As will be convenient also for the next section, we do not store an embedding~$\Pi$ of a graph~$G$ on a surface~$\surf$ by its overlay with the triangulation, as was done in the previous section, but we forget the triangulation.  In other words, we have to store the combinatorial map corresponding to~$\Pi$, but taking into account the fact that $\Pi$ is not necessary cellular:  We need to store, for each face of the embedding, whether it is orientable or not, and a pointer to an edge of each of its boundary components (with some orientation information).  Such a data structure is known under the name of \emph{extended combinatorial map}~\cite[Section~2.2]{cm-tgis-14} (only orientable surfaces were considered there, but the data structure readily extends to non-orientable surfaces).

\subsection{Reduction to connected surfaces}

We first build intermediary EEPs over connected surfaces. Let \emphdef{\EEPCon} be the problem \EEP, restricted to instances $(G,H,\Pi,\surf)$ where $\surf$ is a surface (connected and without boundary) and $H$ intersects every connected component of~$G$.

\begin{lemma}\label{L:connected}
Any instance of \EEPSurf$(n,m,c)$ reduces to $O(m^c)$ instances of \EEPCon$(n,m+c,c)$.
\end{lemma}
More precisely (and this is a fact that will be useful to prove that \Embed{} is in NP, see Theorem~\ref{T:np}), any instance of \EEPSurf{}$(n,m,c)$ is equivalent to the disjunction (OR) of $O(m^c)$ instances, each of them being the conjunction (AND) of $O(c)$ instances of \EEPCon$(n,m+O(c),c)$.

\begin{proof}[Sketch of proof]
  \ifarxiv See Appendix~\ref{A:cellularise} for details. \fi
  We can embed each connected component of~$G$ that is planar and disjoint from~$H$ anywhere.  There remains $O(c)$ connected components of~$G$ that are disjoint from~$H$.  For each of these, we choose a vertex, and we guess in which face of~$\Pi$ it belongs.  We then know which connected component of the surface each connected component of~$G$ is mapped into.
\end{proof}

\subsection{The induction}

The strategy for the proof of Proposition \ref{P:cell} is as follows.  For each EEP $(G',H',\Pi',\surf')$ from the previous lemma, we will extend~$H'$ to make it cellular, by adding either paths connecting two boundary components of a face of~$H'$, or paths with endpoints on the same boundary component of a face of~$H'$ in a way that the genus of the face decreases.  We first define an invariant that will allow to prove that this process terminates.

Let $\Pi$ be an embedding of a graph~$H$ on a surface~$\surf$.  The \emphdef{cellularity defect} of $(H,\Pi,\surf)$ is the non-negative integer
\[\emphdef{cd}(H,\Pi,\surf):= \underset{f\text{ face of }\Pi}{\sum} \text{genus}(f) + \underset{f\text{ face of }\Pi}{\sum}(\text{number of boundaries of }f-1).\]

Some obvious remarks: $\Pi$ can contain isolated vertices; by convention, each of them counts as a boundary component of the face of~$\Pi$ it lies in.  With this convention, every face of~$H$ has at least one boundary component, except in the very trivial case where $G$ is empty.  This implies that $\Pi$ is a cellular embedding if and only if $cd(H,\Pi,\surf)=0$.

The following lemma reduces an EEP to EEPs with a smaller cellularity defect.
\begin{lemma}\label{L:cdlowering}
  Any instance of \EEPCon$(n,O(n),c)$ reduces to $O(n^4)$ instances \\ $(G',H',\Pi',\surf)$ of \EEPCon$(n+O(1),O(n),c)$ where $cd(H',\Pi',\surf)<cd(H,\Pi,\surf)$.

  The reduction does not depend on the size of~$H$; furthermore, the new graph~$G'$ is obtained from the old one by adding exactly one edge and no vertex.
\end{lemma}

Admitting Lemma \ref{L:cdlowering}, the proof of Proposition \ref{P:cell} is straightforward:

\begin{proof}[Proof of Proposition \ref{P:cell}]
We first apply Lemma \ref{L:connected}, obtaining $O(m^c)$ instances of \EEPCon$(n,m+c,c)$.  To each of these EEPs, we apply recursively Lemma~\ref{L:cdlowering} until we obtain cellular EEPs.  The cellularity defect of the initial instance $(G,H,\Pi,\surf)$ is $O(m+c)$, being at most the genus of $\surf$ plus $2m$, because each boundary component of a face of~$\Pi$ is incident to at least one edge of~$H$ (and each edge accounts for at most two boundary components in this way) or to one isolated vertex of~$H$.  Thus, the number of instances of \EEPCell{} at the bottom of the recursion tree is $(n+m+c)^{O(m+c)}$, in which the size of the graph is $O(n+m+c)$ and the surface has at most $c$ simplices.
\end{proof}

\subsection{Proof of Lemma~\ref{L:cdlowering}}

There remains to prove Lemma \ref{L:cdlowering}.  The proof uses some standard notions in surface topology, homotopy, and homology; we refer to textbooks and surveys~\cite{mt-gs-01,s-ctcgt-93,c-ctgs-17}.  We only consider homology with $\mathbb{Z}/2\mathbb{Z}$ coefficients.

Let $f$ be a surface with a single boundary component and let $p$ be a path with endpoints on the boundary of~$f$.  If we contract this boundary component to a single point, the path~$p$ becomes a loop, which can be null-homologous or non-null-homologous.  We employ the same adjectives null-homologous and non-null-homologous for the path~$p$.  Recall that, if $p$ is simple, it separates~$f$ if and only if it is null-homologous.  The reversal of a path~$p$ is denoted by~$\bar{p}$.  The concatenation of two paths $p$ and~$q$ is denoted by~$p\cdot q$.
  
\begin{lemma}\label{L:3paths}
  Let $f$ be a surface with boundary, let $a$ be a point in the interior of~$f$, and let $a_1$, $a_2$, and~$a_3$ be points on the boundary of~$f$.  For each~$i$, let $p_i$ be a path connecting $a_i$ to~$a$.  Let $r_1=p_2\cdot\bar p_3$, $r_2=p_3\cdot\bar p_1$, and $r_3=p_1\cdot\bar p_2$.
  Let $P$ be a possible property of the paths $r_i$, among the following ones:
  \begin{itemize}
      \item ``the endpoints of~$r_i$ lie on the same boundary component of~$f$'';
      \item ``$r_i$ is null-homologous'' (if $f$ has a single boundary component).
  \end{itemize}
  Then the following holds:  If both $r_1$ and~$r_2$ have property~$P$, then so does $r_3$.
\end{lemma}
\begin{proof}
  This is a variant on the
  \emph{3-path condition} from Mohar and Thomassen~\cite[Section~4.3]{mt-gs-01}.
  The first item is immediate.  The second one follows from the fact that homology is an algebraic condition: The concatenation of two null-homologous paths is null-homologous, and removing subpaths of the form $q\cdot\bar q$ from a path does not affect homology.
\end{proof}

\begin{proof}[Proof of Lemma~\ref{L:cdlowering}]
  Since $cd(H,\Pi,\surf)\ge1$, there must be a face~$f$ of~$H$ with either (1) several boundary components, or (2) a single boundary component but positive genus.  We will consider each of these cases separately, but first introduce some common terminology.
  
  Let $F$ be an arbitrary spanning forest of~$G-E(H)$ rooted at $V(H)$.  This means that $F$ is a subgraph of $G-E(H)$ that is a forest with vertex set $V(G)$ such that  each connected component of~$F$ contains exactly one vertex of~$V(H)$, its \emph{root}.  The algorithm starts by computing an arbitrary such forest~$F$ in linear time.
  
  For each vertex~$u$ of~$G$, let $r(u)$ be the unique root in the same connected component of~$F$ as~$u$, and let $F(u)$ be the unique path connecting $u$ to~$r(u)$. If $u$ and~$v$ are two vertices of~$G$, let $G_{uv}$ be the graph obtained from~$G$ by adding one edge, denoted $uv$, connecting $u$ and~$v$.  (This may be a parallel edge if $u$ and~$v$ were already adjacent in~$G$, but in such a situation when we talk about edge $uv$ we always mean the new edge.) Let $F(uv)$ be the unique path between $u$ and~$v$ in~$G$ that is the concatenation of $\overline{F(u)}$, edge $uv$, and~$F(v)$.  
  
  \noindent\textbf{Case~1: $f$ has several boundary components.}  

  Assume that $(G,H,\Pi,\surf)$ has a solution~$\Gamma$.  We claim that, \emph{for some vertices $u$ and~$v$ of~$G$, the embedding~$\Gamma$ extends to an embedding of~$G_{uv}$ in which the image of the path $F(uv)$ lies in~$f$ and connects two distinct boundary components of~$f$}. 
  
  Indeed, let $c$ be a curve drawn in~$f$ connecting two different boundary components of~$f$.  We can assume that it intersects the boundary of~$f$ exactly at its endpoints, at vertices of~$H$.  We can deform~$c$ so that it intersects~$\Gamma$ only at the images of vertices, and never in the relative interior of an edge.  We can, moreover, assume that $c$ is simple (except perhaps that its endpoints coincide on~$\surf$).    Let $v_1,\ldots,v_k$ be the vertices of~$G$ encountered by~$c$, in this order. We denote by $c[i,j]$ the part of~$c$ between vertices $v_i$ and~$v_j$. 
  We claim that, for some~$i$, we have that $\overline{F(v_i)}\cdot c[i,i+1]\cdot F(v_{i+1})$ connects two different boundary components of~$f$:  Otherwise, by induction on~$i$, applying the first case of Lemma~\ref{L:3paths}
  to the three paths $\overline{c[1,i]}$, $F(v_i)$, and $c[i,i+1]\cdot F(v_{i+1})$, we would have that, for each~$i$, $c[1,i]\cdot F(v_i)$ has its endpoints on the same boundary component of~$f$, which is a contradiction for $i=k$ (for which the curve is~$c$).  So let $i$ be such that $\overline{F(v_i)}\cdot c[i,i+1]\cdot F(v_{i+1})$ connects two different boundary components of~$f$; letting $u=v_i$ and~$v=v_{i+1}$, and embedding edge $uv$ as $c[i,i+1]$, gives the desired embedding of~$G_{uv}$.  This proves the claim.
  
  The strategy now is to guess the vertices $u$ and~$v$ and the way the path $F(uv)$ is drawn in~$f$, and to solve a set of EEPs $(G_{uv}, H\cup F(uv), \Pi',\surf)$ where $\Pi'$ is chosen as an appropriate extension of~$\Pi$.  \emph{Let us first assume that $f$ is orientable.}  One subtlety is that, given $u$ and~$v$, there can be several essentially different ways of embedding $F(uv)$ inside~$f$, if there is more than one occurrence of $r(u)$ and~$r(v)$ on the boundary of~$f$. So we reduce our EEP to the following set of EEPs:  For each choice of vertices $u$ and~$v$ of~$G$, and each occurrence of~$r(u)$ and~$r(v)$ on the boundary of~$f$, we consider the EEP $(G_{uv}, H\cup F(uv),\Pi',\surf)$ where $\Pi'$ extends $\Pi$ and maps $F(uv)$ to an arbitrary path in~$f$ connecting the chosen occurrences of $r(u)$ and~$r(v)$ on the boundary of~$f$.
  
  It is clear that, if one of these new EEPs has a solution, the original EEP has a solution.  Conversely, let us assume that the original EEP $(G,H,\Pi,\surf)$ has a solution; we now prove that one of these new EEPs has a solution.  By our claim above, for some choice of $u$ and~$v$, some EEP $(G_{uv},H\cup F(uv),\Pi'',\surf)$ has a solution, for some $\Pi''$ mapping $F(uv)$ inside~$f$ and connecting different boundary components of~$f$.  In that mapping, $F(uv)$ connects two occurrences of $r(u)$ and~$r(v)$ inside~$f$.  We prove that, for these choices of occurrences of $r(u)$ and~$r(v)$, the corresponding EEP described in the previous paragraph, $(G_{uv}, H\cup F(uv),\Pi',\surf)$, has a solution as well.  These two EEPs are the same except that the path $F(uv)$ may be drawn differently in $\Pi'$ and~$\Pi''$, although they connect the same occurrences of $r(u)$ and~$r(v)$ on the boundary of~$f$.  Under~$\Pi'$, the face~$f$ is transformed into a face~$f'$ that has the same genus and orientability character as~$f$, but one boundary component less.  The same holds, of course, for~$\Pi''$.  Moreover, the ordering of the vertices on the boundary components of the new face is the same in $\Pi'$ and~$\Pi''$.
  Thus, there is a homeomorphism~$h$ of~$f$ that keeps the boundary of~$f$ fixed pointwise and such that $h\circ\Pi''|_{F(uv)}=\Pi'|_{F(uv)}$.  This homeomorphism, extended to the identity outside~$f$, maps any solution of $(G_{uv},H\cup F(uv),\Pi'',\surf)$ to  a solution of $(G_{uv},H\cup F(uv),\Pi',\surf)$, as desired.
  
  It also follows from the previous paragraph that the cellularity defect decreases by one.
  To conclude this case, we note that the number of new EEPs is $O(n^4)$: indeed, there are $O(n)$ possibilities for the choice of $u$ (or~$v$), and $O(n)$ possibilities for the choice of the occurrence of $r(u)$ (or $r(v)$) on the boundary of~$f$.
  
  \emph{If $f$ is non-orientable}, the same argument works, except that there are two possibilities for the cyclic ordering of the vertices along the new boundary component of the new face: If we walk along one of the boundary components of~$f$ (in an arbitrary direction), use~$p$, and walk along the other boundary component of~$f$, we do not know in which direction this second boundary component is visited.  So we actually need to consider two EEPs for each choice of $u$, $v$, and occurrences of $r(u)$ and~$r(v)$, instead of one.  The rest  is unchanged.
  
  \noindent\textbf{Case~2: $f$ has a single boundary component and positive genus.}
  The proof is very similar to the previous case, the main difference being that, instead of paths in~$f$ connecting different boundary components of~$f$, we now consider paths in~$f$ that are non-null-homologous.
  \ifarxiv See Appendix~\ref{A:cellularise}.\fi
\end{proof}

\section{Solving a cellular EEP on a surface}\label{S:cellular_solving}
\begin{proposition}\label{P:bojan}
  There is a function $f:\NN\to\NN$ such that every instance of \EEPCell$(n,O(n),c)$ can be solved in time $f(c)\cdot O(n)$.
\end{proposition}
\begin{proof}
  This is essentially the main result of Mohar~\cite{m-ltaeg-99}.  The algorithm in~\cite{m-ltaeg-99} makes reductions to even more specific EEPs, and one feature that is needed for bounding the number of new instances is that the embedded subgraph~$H$ satisfies some connectivity assumptions.
  \ifarxiv Our Appendix~\ref{A:cellular_solving} provides more details.
  \fi
\end{proof}

\section{Proof of Theorems~\ref{T:np} and~\ref{T:main}}\label{S:main_proof}
We can finally prove our main theorems.
First, let us prove that we have an algorithm with complexity $f(c)\cdot n^{O(c)}$:
\begin{proof}[Proof of Theorem~\ref{T:main}]
  This immediately follows from Propositions~\ref{P:3-book}, \ref{P:pure}, \ref{P:surf}, \ref{P:cell}, and~\ref{P:bojan}.
  \ifarxiv See Appendix~\ref{A:main_proof} for details.\fi
\end{proof}

Finally, we prove that the \Embed{} problem is NP-complete:
\begin{proof}[Proof of Theorem~\ref{T:np}]
  \ifarxiv We give a detailed proof in Appendix~\ref{A:main_proof}.  
  \fi
  The problem \Embed{} is NP-hard because deciding whether an input graph embeds into an input surface is NP-hard~\cite{t-ggpnc-89}.  It is in NP because (assuming $\complex$ contains no 3-books),  
  a certificate that an instance is positive is given by a certificate that some instance of \EEPSurf{} is positive (see the proof of 
  Proposition~\ref{P:surf}).  Such an instance $(G,H,\Pi,\surf)$ has a certificate, given by the combinatorial map of a supergraph $G'$ of~$G$ cellularly embedded on~$\surf$ (see Section~\ref{S:cellularise}), that can be checked in polynomial time.
\end{proof}

\bibliographystyle{plain}
\bibliography{bib}

\ifarxiv
\appendix
\appendix

\section{Omitted proof from Section~\ref{S:impure}}\label{A:impure}

\begin{proof}[Proof of Lemma~\ref{L:reduc-pure}]
  Formally, we describe a set of transformations on $\complex$, $G$, and~$f$.  The invariant is that they preserve the existence or non-existence of an embedding of~$G$ into~$\complex$ respecting~$f$.
 
 \def\stp#1{\emph{Step~#1.\quad}}%
  \stp{1} We start by dissolving all the degree-two vertices of~$G$ that are not in the image of~$f$.  (We recall that \emph{dissolving a degree-two vertex} $v$, incident to edges $vv_1$ and~$vv_2$, means removing $v$ and replacing $vv_1$ and~$vv_2$ with a single edge $v_1v_2$.)  It is clear that the original graph~$G$ has an embedding on~$\complex$ respecting~$f$ if and only if the new graph (still called~$G$) has an embedding on~$\complex$ respecting~$f$.

  \smallskip

  \stp{2} If a singular node is not used, then we can remove it without affecting the embeddability of~$G$.  However, removing a vertex from a 2-complex does not yield a 2-complex.  We thus define a 2-complex that has the same properties.  Let $p$ be a singular node of a 2-complex~$\complex$. Let $T$ be the set of triangles and segments of~$\complex$ incident with~$p$, uniquely partitioned into $T_1,\ldots,T_k$, where each $T_i$ is either a cone, a corner, or an isolated segment.  The \emph{withdrawal of $p$ from $\complex$} is the complex obtained by doing the following operation for each $i=1,\ldots,k$:  We first create a new node~$p_i$, and then replace, in each triangle and edge of~$T_i$, the node~$p$ by~$p_i$.

  For every singular node~$p$ of~$\complex$ such that $f(p)=\varepsilon$ and incident to at least two segments, we withdraw~$p$ from~$\complex$.  For each created node $p_i$, we let $f(p_i)=\varepsilon$.  The fact that $G$ embeds, or not, on~$\complex$ respecting~$f$ is preserved: Indeed, if $G$ embeds on the original complex respecting~$f$, then this corresponds to an embedding on the complex obtained by withdrawing~$p$, and avoiding the $p_i$s, thus respecting~$f$; conversely, if $G$ embeds on the complex obtained by withdrawing~$p$, respecting~$f$, then it avoids the $p_i$s, and, after identifying together the $p_i$s to a single point~$p$, this embedding avoids~$p$, and thus respects~$f$.

  \smallskip
  
  \stp{3}  At this point, every singular node~$p$ with $f(p)=\varepsilon$ is incident to exactly one segment, and to no triangle.  For each isolated segment $pq$ of~$\complex$ such that $u:=f(p)$ and~$v:=f(q)$ are both different from~$\varepsilon$, and $G$ contains an edge $uv$ connecting $u$ and~$v$, we remove $uv$ from~$G$ and remove $pq$ from~$\complex$.  We need to prove that this operation does not affect the (non-)existence of an embedding of~$G$ respecting~$f$.  First, assume that, initially, there was an embedding~$\Gamma$ of~$G$ on~$\complex$ respecting~$f$; then either segment $pq$ is used by~$uv$, in which case clearly there is still an embedding after this operation, or edge $uv$ does not use segment~$pq$ at all, in which case we can first modify~$\Gamma$ by embedding edge~$uv$ on segment~$pq$ and by moving the edges on~$pq$ on the space where $uv$ was before, so we are now in the previous case.  Conversely, if after this operation $G$ has an embedding on~$\complex$ respecting~$f$, trivially it is also the case before.  
  
  \smallskip

  \stp{4}  For each isolated segment $pq$ of~$\complex$ such that $u:=f(p)$ and~$v:=f(q)$ are both different from~$\varepsilon$, but $G$ contains no edge connecting $u$ and~$v$, we do the following.  In $\complex$, we remove~$pq$ and add a new segment $p'q'$ where $p'$ and~$q'$ are new nodes; we also extend $f$ by letting $f(p')=f(q')=\varepsilon$.  Finally, if $G$ contains at least one edge of the form $ux$, where $x$ has degree one and is not in the image of~$f$, we remove a single such edge; similarly, if it contains at least one edge of the form $vy$, where $y$ has degree one and is not in the image of~$f$, we remove a single such edge.  This operation does not affect our invariant, for similar reasons; for example, if initially there was an embedding of~$G$ respecting~$f$, then segment~$uv$ can only contain edges of~$G$ that are themselves connected components of~$G$, which we can re-embed on the new segment~$p'q'$, or edges of the form $ux$ or $vy$, where $x$ and~$y$ have degree one and are not in the image of~$f$.
  
    \smallskip

  \stp{5} For each isolated segment $pq$ where $u:=f(p)$ is different from~$\varepsilon$ but $f(q)=\varepsilon$, we remove~$pq$ from~$\complex$ and add a new segment $p'q'$ where $p'$ and~$q'$ are new nodes; we also extend $f$ by letting $f(p')=f(q')=\varepsilon$.  Finally, if $G$ contains at least one edge of the form $ux$, where $x$ has degree one and not in the image of~$f$, we remove a single such edge.  The invariant is preserved, for reasons similar to the previous case.
  
    \smallskip

  \stp{6}  Now, every segment of the complex (still called~$\complex$) is incident to one or two triangles, except perhaps some segments that are themselves connected components of~$\complex$ and whose endpoints are not in the domain of~$f$.  If there is at least one such segment, we remove all of them from~$\complex$, and remove all edges $uv$ from~$G$ that are themselves connected components of~$G$ and such that $u$ and~$v$ are not in the image of~$f$; as above, this does not affect whether $G$ embeds into~$\complex$ respecting~$f$.
  
    \smallskip
  
  \stp{7}  Finally, for each node~$p$ of~$\complex$ incident to no segment, such that $u:=f(p)$ is different from~$\varepsilon$, we do the following:  If $u$ has degree zero, we remove~$p$ and~$u$; otherwise, we immediately return that $G$ does not embed into~$\complex$ respecting~$f$.  For each node~$p$ of~$\complex$ that is incident to no segment and such that $f(p)=\varepsilon$, we remove~$p$, and remove a single degree-zero vertex of~$G$ not in the image of~$f$, if one exists.
  
    \smallskip

  \emph{Conclusion.}\quad Now, $\complex$ has no node that is itself a connected component; each of its segments is incident to one or two triangles.   Also, $f$ maps each singular node of~$\complex$ to a vertex of~$G$.  It may map some other nodes of~$\complex$ to~$\varepsilon$, but such nodes are non-singular and, if an embedding uses them, a slight perturbation will avoid them, so we can remove the nodes~$p$ such that $f(p)=\varepsilon$ from the domain of~$f$ without affecting the result.  Now, we have an EEP as specified in the statement of Proposition~\ref{P:pure}.

  Finally, it is easy to check that, in each of the seven steps above, the numbers of vertices and edges of~$G$ do not increase, and the number of simplices of~$\complex$ increase by at most a multiplicative constant.  Moreover, the domain of~$f$ also does not increase (it increases when we withdraw singular nodes, but the images of the new nodes are~$\varepsilon$, and such nodes are later removed from the domain of~$f$).
\end{proof}

\section{Omitted proofs from Section~\ref{S:pure}}\label{A:pure}

\begin{figure}
\centering
\includegraphics[width=.6\linewidth]{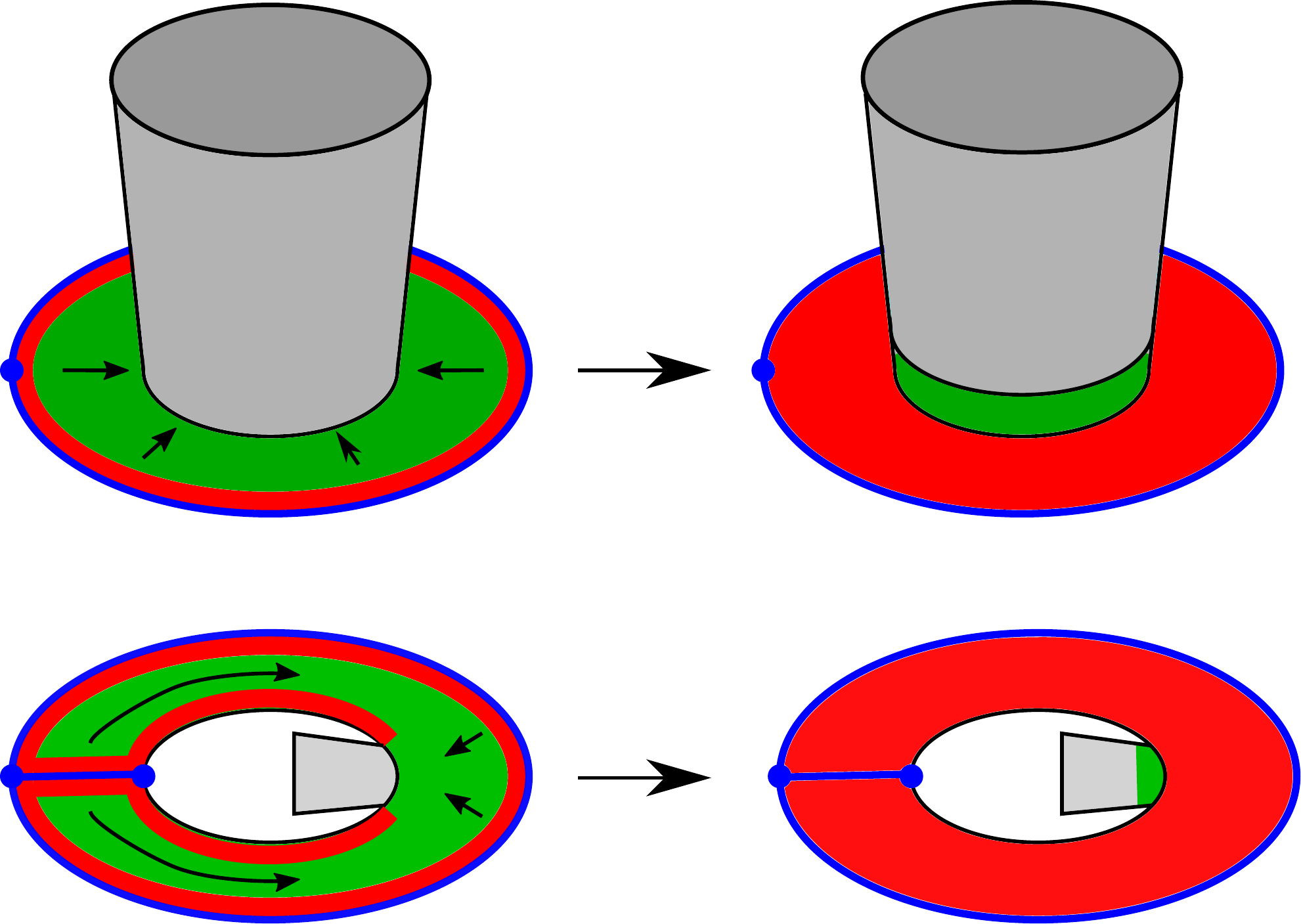}
\caption{In the proof of Lemma~\ref{L:surf-bound}, we push some parts of the graph outside~$S_p$ by an ambient isotopy of~$\surf$ that moves a small enough open disk or annulus (in blue), which is disjoint from $L_p$ and~$E_p$, to a larger part.}
\label{F:homeo}
\end{figure}
\begin{figure}
\centering
\includegraphics[width=.8\linewidth,height=4cm]{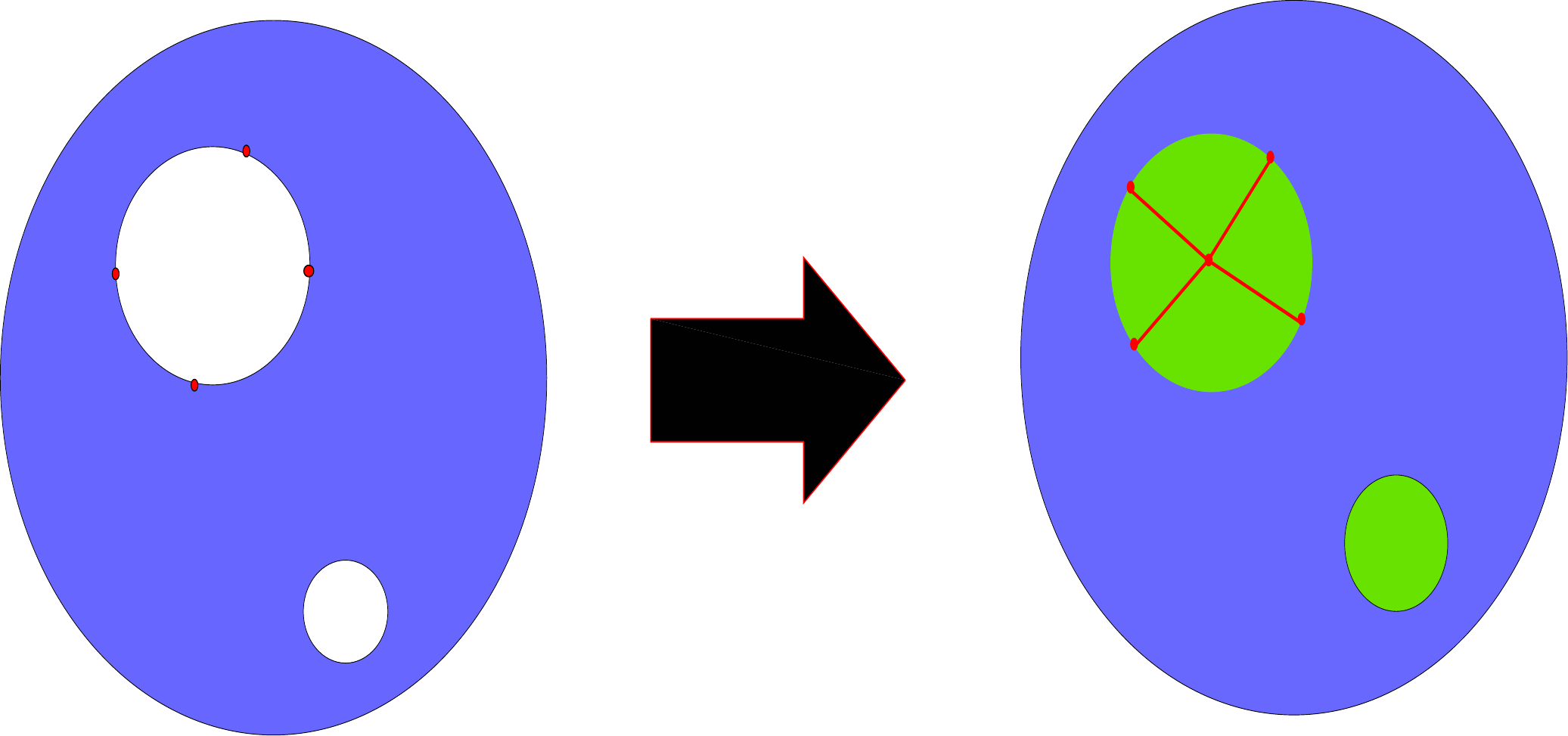}
\caption{The removal of boundary components in the proof of
 Lemma~\ref{L:surf nobound}.  We attach a disk to every boundary component
  of the surface.  Moreover, if there is a boundary component~$c$
  containing at least one vertex of~$H$, say $v_1,\ldots,v_k$, we add
  to~$H$ a new vertex~$v$, mapped inside the corresponding disk, and add,
  in $H$, one edge between $v$ and each of the~$v_i$s, that edge being
  mapped inside the disk.}
\label{F:remove-boundary}
\end{figure}

\begin{proof}[End of proof of Lemma~\ref{L:surf-bound}]
  There remains to prove that the two \EEP{}s constructed are equivalent.
 
  Let us first prove that any solution~$\Gamma$ of the \EEP{} instance $(G,H,\Pi,\complex)$ yields a solution of the \EEP{} instance $(G',H',\Pi',\surf)$.  Let $p$ be a singular node of~$\complex$, and let $v_p$ be the vertex of~$H$ mapped to~$p$.  We need to modify $\Gamma$ locally in the neighorhood~$N_p$ of~$p$ that is removed when transforming $\complex$ to~$\surf$. Without loss of generality, up to an ambient isotopy of~$\Gamma$ that does not move~$H$, we can assume that the image of~$\Gamma$ intersects~$N_p$ exactly in straight line segments having~$p$ as one endpoint.  To build a solution of $(G',H',\Pi',\surf)$, we first remove the part of~$\Gamma$ inside~$N_p$, and we reconnect~$v_p$ to each point of the image of~$\Gamma$ lying on the boundary of~$N_p$, by paths on~$S_p$; this is certainly possible because each point of~$S_p$ not in the image of $E_p\cup L_p$ can be connected to~$v_p$ by a path that does not meet $E_p\cup L_p$ (except at~$v_p$).  Thus, $(G',H',\Pi',\surf)$ has a solution.

\smallskip

  Conversely, let $\Gamma'$ be a solution of the \EEP{} $(G',H',\Pi',\surf)$; we build a solution of $(G,H,\Pi,\complex)$.  As above, let $q_p$ be a point of~$\surf$ that was obtained from a singular node~$p$ of~$\complex$. We now show that we can assume that $\Gamma$ does not enter~$S_p$, except for~$v_p$, $E_p$, $L_p$, and the edges in~$G'$ incident to~$v_p$:

  First, the part of~$\Gamma$ lying in the connected component~$D_p$ of~$S_p$ minus the image of~$L_p$ that is a disk corresponds to a planar subgraph of~$G$, connected to the rest of~$G$ by~$v_p$ only; we can re-embed this planar subgraph of~$G$ in~$S_p\setminus D_p$.  Next, consider a cone at~$p$.  The part of~$S_p$ enclosed by the corresponding loop of~$L_p$ is an annulus; the situation is as on Figure~\ref{F:homeo}, top, and, by an ambient isotopy of~$\surf$, we can push the part of~$\Gamma$ that lies in the annulus outside it, except for those edges touching~$v_p$.  Finally, consider a corner at~$p$, and the annulus that is the part of~$S_p$ enclosed by the corresponding loop in~$L_p$.  The local picture for this part of~$S_p$  is as shown on Figure~\ref{F:homeo}, bottom, and similarly by an ambient isotopy of~$\surf$ we can push the part of~$\Gamma$ on~$S_p$ outside it, except for those edges touching~$v_p$.

  Now, a solution of $(G,H,\Pi,\complex)$ can be obtained by the following procedure, for each of the singular nodes~$p$:  (1) Remove the sphere~$S_p$, together with the image of~$\Gamma$ inside~$S_p$; (2) add the neighborhood~$N_p$ of~$p$; (3) reconnect to~$p$ the points on the image of~$\Gamma$ that lie on the boundary of~$N_p$.
\end{proof}

\begin{proof}[End of proof of Lemma~\ref{L:surf nobound}]
  Clearly, any solution of $(G,H,\Pi,\surf)$ yields a solution of \\$(G',H',\Pi',\surf')$.  Conversely, let~$\Gamma'$ be a solution of $(G',H',\Pi',\surf')$; we build a solution of $(G,H,\Pi,\surf)$.  Let $b$ be a boundary of $\surf$.  There are two cases: 
  \begin{itemize}
  \item If no vertex of $H$ is mapped to~$b$, then $\Pi'$ maps~$H'$ outside the closure of~$D_b$, so, by an ambient isotopy of~$\Gamma'$ that does not move~$H'$, we can push~$\Gamma'$ outside~$D_b$. 
  \item Otherwise, the disk~$D_b$ is split into sectors by the edges incident to~$v_b$; the image of~$H'$ by~$\Pi'$ does not enter the interior of each sector, and intersects the boundary of a sector exactly along the image of~$v_b$ and of two edges incident to~$v_b$.  Thus, by an ambient isotopy that keeps the image of~$H'$ fixed, we can push the image of~$G'$ out of each sector.
  \end{itemize}
  After doing this for every boundary component~$b$, we obtain that the restriction of~$\Gamma'$ to~$G$ is a solution of $(G,H,\Pi,\surf)$.
\end{proof}

\section{Omitted proofs from Section~\ref{S:cellularise}}\label{A:cellularise}

  \begin{proof}[Proof of Lemma~\ref{L:connected}]
  We start by removing the connected components of~$G$ that are planar and disjoint from~$H$.  (Testing planarity takes linear time~\cite{ht-ept-74}.)
  This does not change the solution of the EEP, because such connected components can be embedded on an arbitrarily small planar portion of~$\surf$ (provided $\surf$ is non-empty, but otherwise the original \EEPSurf{} instance can be solved trivially).  So without loss of generality, every connected component of~$G$ disjoint from~$H$ is non-planar.  Let $V_0$ be an arbitrary set of vertices, one per non-planar connected component of~$G$ disjoint from~$H$.  Without loss of generality, the number of vertices in~$V_0$ is at most the genus of~$\surf$, which is at most $c$, because otherwise the initial EEP has no solution.
  
  Let $H':=H\cup V_0$; every connected component of~$G$ intersects~$H'$.  
  For each vertex of~$V_0$, we guess the face of~$\Pi$ it has to be embedded in, and extend $\Pi$ accordingly, by adding the images of~$V_0$ in~$\Pi$; let $\Pi'$ be the resulting embedding of~$H'$.   By the previous paragraph, the number of these guesses is at most $m^c$. It is clear that the initial EEP has a solution if and only if one of these EEPs $(G,H',\Pi',\surf)$ has a solution.
  
  These EEPs are almost of the form announced in the lemma, except that $\surf$ can be disconnected.  However, in any solution of this EEP, we know the connected component of~$\surf$ each connected component of~$G$ has to embed in, because each connected component of~$G$ intersects~$H$.  We can thus reformulate the EEP as the conjunction of several EEPs, one per connected component of~$\surf$.  (Of course, we can discard the connected components of~$\surf$ disjoint from~$H$.)
\end{proof}

\begin{proof}[End of proof of Lemma~\ref{L:cdlowering}]
   There remains to prove Case~2: $f$ has a single boundary component and positive genus.
  
  The proof is very similar to Case~1, the main difference being that, instead of paths in~$f$ connecting different boundary components of~$f$, we now consider paths in~$f$ that are non-null-homologous (if $f$ is orientable) or one-sided (if $f$ is non-orientable).  For brevity we call such paths \emph{essential}.  Observe that Lemma~\ref{L:3paths} is similarly valid for the property ``$r_i$ is one-sided'' (if $f$ has a single boundary component).
  
  Assume that $(G,H,\Pi,\surf)$ has a solution~$\Gamma$.  We claim that, \emph{for some vertices $u$ and~$v$ of~$G$, the embedding~$\Gamma$ extends to an embedding of~$G_{uv}$ in which the image of the path $F(uv)$ lies in~$f$ and is essential}.  The proof is similar in spirit to the corresponding claim in Case~1: We let $c$ be an essential curve in~$f$ intersecting the boundary of~$f$ exactly at its endpoints; we can assume similarly that it is simple and intersects only vertices of~$\Gamma$, in the order $v_1,\ldots,v_k$.  For some~$i$, $\overline{F(v_i)}\cdot c[i,i+1]\cdot F[v_{i+1}]$ must be essential, by induction and by Lemma~\ref{L:3paths}; this gives an embedding of~$G_{uv}$.
  
  \emph{If $f$ is orientable}, we reduce the original EEP to the following set of EEPs: For each choice of vertices $u$ and~$v$ of~$G$, and each occurrence of $r(u)$ and~$r(v)$ on the boundary of~$f$, we consider the EEP $(G_{uv}, H\cup F(uv),\Pi',\surf)$ where $\Pi'$ extends~$\Pi$ and maps $F(uv)$ to an arbitrary path~$p$ in~$f$ connecting the chosen occurrences of~$r(u)$ and~$r(v)$ on the boundary of~$f$ in such a way that $p$ is non-null-homologous.
  As before, the only subtlety is to prove that, if we have two EEPs $(G_{uv},H\cup F(uv),\Pi',\surf)$ and $(G_{uv},H\cup F(uv),\Pi'',\surf)$ such that $F(uv)$ are not embedded exactly in the same way in $\Pi'$ and~$\Pi''$, but are non-null-homologous in~$f$ and connect the same occurrences of~$r(u)$ and~$r(v)$ on the boundary of~$f$, then these EEPs are equivalent.  This follows from the fact that the image of $F(uv)$ in~$\Pi'$ cuts~$f$ into a face that is an orientable surface with two boundary components and with (Euler) genus that of~$f$ minus two (because $F(uv)$ is non-null-homologous in both embeddings, and thus non-separating in both cases).  Moreover, the ordering of the vertices along the boundary components of the new face is the same in both $\Pi'$ and~$\Pi''$.  Thus, as before, there is a homeomorphism~$h$ of~$f$ that keeps the boundary of~$f$ fixed pointwise and such that $h\circ\Pi''|_{F(uv)}=\Pi'|_{F(uv)}$, which similarly concludes.  It also follows that the cellularity defect decreases by two.  The number of these new EEPs is, similarly, $O(n^4)$.
  
  \emph{If $f$ is non-orientable and the genus of~$f$ is either one or even}, then by Euler's formula, a similar argument can be used: Regardless of the way we draw $F(uv)$ as a path $p$ connecting the chosen occurrences of $r(u)$ and~$r(v)$ on the boundary of~$f$ in such a way that $p$ is one-sided, cutting $f$ along~$p$ results in a surface whose topology is uniquely determined (it is a disk if the genus of~$f$ is one, and a non-orientable surface with genus that of~$f$ minus one, if the genus of~$f$ is even).  Moreover, the ordering of the vertices along the boundary of the new face is uniquely determined.  The cellularity defect decreases by one, and the same argument as above concludes.
 
  Finally, \emph{if $f$ is non-orientable and its genus is odd and at least three}, then cutting~$f$ along such a one-sided path~$p$ results in a surface in which the ordering of the vertices along the single boundary component is uniquely determined, but this surface, with genus that of~$f$ minus one, can be orientable or not.  Thus, for each choice of vertices $u$ and~$v$ of~$G$, and each occurrence of $r(u)$ and~$r(v)$ on the boundary of~$f$, we actually need to consider two EEPs, one in which $F(uv)$ is mapped to a path that cuts~$f$ into an orientable surface, and one in which $F(uv)$ is mapped to a path that cuts~$f$ into a non-orientable surface.  The rest of the argument is unchanged.
\end{proof}

\section{Omitted proof from Section~\ref{S:cellular_solving}}\label{A:cellular_solving}

When solving EEPs on surfaces, a useful property of the embedded subgraph is the following one.  We say that a subgraph~$H$ of~$G$ has \emph{property~(E)} if $H$ has no local bridges.  This means that every path $P$ of~$H$, all of whose vertices have degree at most two in~$H$, is an induced path in~$G$, and every connected component of $G-V(H)$ is adjacent to a vertex in $V(H)\setminus V(P)$.  In fact, what we need is a weaker version of property~(E) where we first prescribe a subset~$V_0$ of vertices of~$H$, where $V_0$ contains all vertices whose degree in~$H$ is different from two, and possibly a constant number of vertices whose degree in~$H$ is equal to two.  Then we say that $H$ has property~(E) with respect to~$V_0$ if the above property holds for every path~$P$ that is disjoint from~$V_0$.
  
\begin{proof}[Proof of Proposition~\ref{P:bojan}]
  This is essentially the main result from~\cite{m-ltaeg-99}.  The algorithm from~\cite{m-ltaeg-99} first reduces the problem to an instance $(G',H',\Pi',\surf)$ such that $H'$ satisfies property~(E) with respect to a subset $V_0$ that contains all vertices of~$H$ of degree different from two (all these are also vertices in~$H'$ of degree different from two) plus a constant number of vertices of degree two in~$H'$, where this constant number is bounded from above in terms of the genus of~$\surf$, which is itself bounded from above by~$c$.  For this purpose, \cite{m-ltaeg-99} relies on another paper~\cite{jmm-elb-97}.  After achieving this property, the paper~\cite{m-ltaeg-99} reduces the EEP to a constant number (where the constant depends on~$c$) of ``simple'' extension problems \cite[Section~4]{m-ltaeg-99}, which are then solved in~\cite[Theorem~5.4]{m-ltaeg-99}.
\end{proof}

\section{Omitted proofs from Section~\ref{S:main_proof}}\label{A:main_proof}

\begin{proof}[End of proof of Theorem~\ref{T:main}]
  Consider an instance of~\Embed$(n,c)$.  
  Proposition~\ref{P:3-book} allows to discard the 2-complexes containing a 3-book.  Proposition~\ref{P:pure} reduces the problem to $(cn)^{O(c)}$ instances of \EEPSing$(cn,c,O(c))$.  Proposition~\ref{P:surf} reduces each such instance into an instance of \EEPSurf$(O(cn),O(c),O(c))$.  Proposition~\ref{P:cell} reduces that instance into $O(cn)^{O(c)}$ instances of \EEPCell$(O(cn),O(cn),O(c))$.  Proposition~\ref{P:bojan} shows that each such instance can be solved in time $f(O(c))\cdot O(cn)$.
  \end{proof}
  
\begin{proof}[Proof of Theorem~\ref{T:np}]
  Let us first prove that the problem is NP-hard.  The following problem \GraphGenus{} is NP-hard: Given a graph~$G$ and an integer~$g$, decide whether $G$ embeds on the orientable surface of genus~$g$~\cite{t-ggpnc-89}.  This almost immediately implies that \Embed{} is NP-hard; the only subtlety is that in \GraphGenus{}, $g$ is specified in binary, thus more compactly than a triangulated surface of genus~$g$ (and thus $\Omega(g)$ triangles).  To be very precise, given an instance $(G,g)$ of \GraphGenus{}, we transform it in polynomial time into an equivalent instance of \Embed{} as follows:  If $G$ has at most $g$ edges, then we transform it into a constant-size positive instance of~\Embed{} (every graph with $g$ edges embeds on the orientable surface of genus~$g$); otherwise, we consider the instance $(G,\complex)$ where $\complex$ is a 2-complex that is an orientable surface of genus~$g$; since $G$ has at least $g$ edges, the transformation takes polynomial time in the size of $(G,g)$.

  We now prove that the problem \Embed{} belongs to NP.  The case where $\complex$ contains a 3-book is trivial; let us assume that it is not the case.  The proof of Proposition~\ref{P:surf} shows that an \Embed{} instance is positive if and only if at least one instance of~\EEPSurf{}, among $(cn)^{O(c)}$ of them, is positive.  The certificate indicates which of these instances is positive (this requires a polynomial number of bits), together with a certificate that this instance is indeed positive (see below).  To check this certificate, the algorithm builds the corresponding instance of~\EEPSurf{} (as done in Section~\ref{S:pure}---this takes polynomial time) and checks the certificate.
  
  Here is a way to provide a certificate for an instance of \EEPSurf.  In Section~\ref{S:cellularise}, we have proved that, if we have an instance $(G,H,\Pi,\surf)$ of~\EEPSurf, then there exists a cellular embedding~$\Gamma'$ (in the form of a combinatorial map) of a graph~$G'$ containing~$G$, and such that $\Gamma'$ extends~$\Pi$.  Moreover, $G'$ is obtained from~$G$ by adding a number of edges that is $O(c)$, where $c$ is the size of the original complex.  (Recall that in the instance of \EEPSurf, the size of~$H$ is $O(c)$.)  The cellular embedding~$\Gamma'$ of~$G'$, given as a combinatorial map, is the certificate that $(G,H,\Pi,\surf)$ is positive:  Given $(G,H,\Pi,\surf)$ and this certificate, we can in polynomial time check that $G'$ contains~$G$, that the restriction of~$\Gamma'$ to~$H$ is indeed~$\Pi$, and that the combinatorial map of~$\Gamma'$ is indeed an embedding on~$\surf$.
\end{proof}

\fi
\end{document}